\documentclass{birkmult}

\usepackage{mathrsfs}

\title{A Quantum Dot with Impurity in the Lobachevsky Plane}

\author[V. Geyler]{V.~Geyler}
\address{Department of Mathematics\\
  Mordovian State University\\
  Saransk, Russia}

\author{P.~\v{S}\v{t}ov\'\i\v{c}ek}
\address{Department of Mathematics\\
 Faculty of Nuclear Sciences\\
 Czech Technical University\\
 Prague, Czech Republic}
\email{stovicek@kmalpha.fjfi.cvut.cz}

\author{M.~Tu\v{s}ek}
\address{Department of Mathematics\\
 Faculty of Nuclear Sciences\\
 Czech Technical University\\
 Prague, Czech Republic}
\email{tusekmat@km1.fjfi.cvut.cz}

\newtheorem{thm}{Theorem}[section]

\newtheorem{prop}[thm]{Proposition}
\theoremstyle{definition}

\theoremstyle{remark}
\newtheorem*{rem*}{Remark}

\numberwithin{equation}{section}

\newcommand{\ud}{\mathrm{d}}
\newcommand{\parc}[2]{\frac{\partial #1}{\partial #2}}
\newcommand{\abs}[1]{\vert #1 \vert}
\newcommand{\lsz}{\left\lbrace }
\newcommand{\psz}{\right\rbrace }
\newcommand{\Gz}{\mathcal{G}_{z}}
\newcommand{\ena}[1]{\mathrm{e}^{#1}}

\newcommand{\N}{\mathbb{N}}
\newcommand{\C}{\mathbb{C}}
\newcommand{\R}{\mathbb{R}}
\newcommand{\Z}{\mathbb{Z}}

\newcommand{\Dom}{\mathop\mathrm{Dom}\nolimits}
\newcommand{\diag}{\mathop\mathrm{diag}\nolimits}
\newcommand{\spec}{\mathop\mathrm{spec}\nolimits}

\begin{document}

\begin{abstract}
  The curvature effect on a quantum dot with impurity is investigated.
  The model is considered on the Lobachevsky plane. The confinement
  and impurity potentials are chosen so that the model is explicitly
  solvable. The Green function as well as the Krein $Q$-function are
  computed.
\end{abstract}


\maketitle

\section{Introduction}

Physically, quantum dots are nanostructures with a charge carriers
confinement in all space directions. They have an atom-like energy
spectrum which can be modified by adjusting geometric parameters of
the dots as well as by the presence of an impurity. Thus the study of
these dependencies may be of interest from the point of view of the
nanoscopic physics.

A detailed analysis of three-dimensional quantum dots with a
short-range impurity in the Euclidean space can be found in
\cite{bgl}. Therein, the harmonic oscillator potential was used to
introduce the confinement, and the impurity was modeled by a point
interaction ($\delta$-potential). The starting point of the analysis
was derivation of a formula for the Green function of the unperturbed
Hamiltonian (i.e., in the impurity free case), and application of the
Krein resolvent formula jointly with the notion of the Krein
$Q$-function.

The current paper is devoted to a similar model in the hyperbolic
plane. The nontrivial hyperbolic geometry attracts regularly
attention, and its influence on the properties of quantum-mechanical
systems has been studied on various models (see, for example,
\cite{com,aco,lisovyy}). Here we make use of the same method as in
\cite{bgl} to investigate a quantum dot with impurity in the
Lobachevsky plane. We will introduce an appropriate Hamiltonian in a
manner quite analogous to that of \cite{bgl} and derive an explicit
formula for the corresponding Green function. In this sense, our model
is solvable, and thus its properties may be of interest also from the
mathematical point of view.

During the computations to follow, the spheroidal functions appear
naturally. Unfortunately, the notation in the literature concerned
with this type of special functions is not yet uniform (see, e.g.,
\cite{be} and \cite{me}). This is why we supply, for the reader's
convenience, a short appendix comprising basic definitions and results
related to spheroidal functions which are necessary for our approach.

\section{A quantum dot with impurity in the Lobachevsky plane}

\subsection{The model}

Denote by $(\varrho,\phi)$, $0<\varrho<\infty$, $0\leq\phi<2\pi$, the
geodesic polar coordinates on the Lobachevsky plane. Then the metric
tensor is diagonal and reads
\begin{equation*}
  (g_{ij}) = \diag\!\left(1,a^2\sinh^2\frac{\varrho}{a}\right)
\end{equation*}
where $a$, $0<a<\infty$, denotes the so called curvature radius which
is related to the scalar curvature by the formula $R=-2/a^2$.
Furthermore, the volume form equals
$dV=a\sinh(\varrho/a)\ud\varrho\wedge\ud\phi$. The Hamiltonian for a
free particle of mass $m=1/2$ takes the form
\begin{equation*}
  H^{0} = -\left(\Delta_{LB}+\frac{1}{4a^{2}}\right)
  = -\frac{1}{\sqrt{g}}\frac{\partial}{\partial{x^{i}}}\sqrt{g}g^{ij}
  \frac{\partial}{\partial{x^{j}}}-\frac{1}{4a^{2}}
\end{equation*}
where $\Delta_{LB}$ is the Laplace-Beltrami operator and
$g=\det{}g_{ij}$. We have set $\hbar=1$.

The choice of a potential modeling the confinement is ambiguous. We
naturally require that the potential takes the standard form of the
quantum dot potential in the flat limit ($a\to\infty$). This is to say
that, in the limiting case, it becomes the potential of the isotropic
harmonic oscillator $V_\infty=\frac{1}{4}\omega^2\varrho^2$. However,
this condition clearly does not specify the potential uniquely. Having
the freedom of choice let us discuss the following two possibilities:
\begin{eqnarray}
  &\mathrm{a)}\quad V_{a}(\varrho)
  = \frac{1}{4}\,a^{2}\omega^{2}\tanh^2\frac{\varrho}{a}, \\
  &\mathrm{b)}\quad U_{a}(\varrho)
  = \frac{1}{4}\,a^{2}\omega^{2}\sinh^2\frac{\varrho}{a}.
\end{eqnarray}

Potential $V_{a}$ is the same as that proposed in \cite{ran} for the
classical harmonic oscillator on the Lobachevsky plane. With this
choice, it has been demonstrated in \cite{ran} that the model is
superintegrable, i.e., there exist three functionally independent
constants of motion. Let us remark that this potential is bounded, and
so it represents a bounded perturbation to the free Hamiltonian. On
the other hand, the potential $U_{a}$ is unbounded. Moreover, as shown
below, the stationary Schr\"odinger equation for this potential leads,
after the partial wave decomposition, to the differential equation of
spheroidal functions. The current paper concentrates exclusively on
case b).

The impurity is modeled by a $\delta$-potential which is introduced
with the aid of self-adjoint extensions and is determined by boundary
conditions at the base point. We restrict ourselves to the case when
the impurity is located in the center of the dot ($\varrho =0$). Thus
we start from the following symmetric operator:
\begin{equation}
  \label{ham}
  \begin{split}
    & \begin{split}H
      = -\left(\parc{^2}{\varrho^2}
        +\frac{1}{a}\coth\!\left(\frac{\varrho}{a}\right)
        \parc{}{\varrho}+\frac{1}{a^2}
        \sinh^{-2}\!\left(\frac{\varrho}{a}\right)
        \parc{^2}{\phi^2}+\frac{1}{4a^2}\right)
      +\frac{1}{4}\,a^{2}\omega^{2}
      \sinh^2\!\left(\frac{\varrho}{a}\right),
    \end{split} \\
    & \Dom(H) = C^{\infty}_{0}((0,\infty)\times S^1)
    \subset L^{2}\left((0,\infty)\times S^1,
      a\,\sinh\!\left(\frac{\varrho}{a}\right)
      \ud\varrho\,\ud\phi\right).
  \end{split}
\end{equation}

\subsection{Partial wave decomposition}

Substituting $\xi=\cosh(\varrho/a)$ we obtain
\begin{equation}
  \label{Htilde}
  \begin{split}
    & H = \frac{1}{a^2}\left[(1-\xi^{2})\parc{^2}{\xi^{2}}
      -2\xi\parc{}{\xi}+(1-\xi^2)^{-1}\parc{^{2}}{\phi^{2}}
      + \frac{a^{4}\omega^{2}}{4}(\xi^{2}-1)-\frac{1}{4} \right]
    =: \frac{1}{a^{2}}\tilde{H}, \\
    & \Dom(H) = C^{\infty}_{0}((1,\infty)\times S^1)
    \subset L^{2}\left((1,\infty)\times S^1,
      a^{2}\ud\xi\,\ud\phi\right).
  \end{split}
\end{equation}
Using the rotational symmetry which amounts to a Fourier transform in
the variable $\phi$, $\tilde{H}$ may be decomposed into a direct sum
as follows
\begin{equation*}
  \begin{split}
    &\tilde{H} = \bigoplus\limits_{m=-\infty}^{\infty}\tilde{H}_{m},\\
    &\begin{split}
      \tilde{H}_{m}
      & = -\parc{}{\xi}\left((\xi^{2}-1)\,\parc{}{\xi}\right)
      +\frac{m^2}{\xi^2-1}+\frac{a^{4}\omega^{2}}{4}\,(\xi^{2}-1)
      -\frac{1}{4}\,,
    \end{split}\\
    & \Dom(\tilde{H}_{m})=C^{\infty}_{0}(1,\infty)
    \subset L^{2}((1,\infty),\ud\xi).
  \end{split}
\end{equation*}
Note that $\tilde{H}_{m}$ is a Sturm-Liouville operator.

\begin{prop}
  $\tilde{H}_{m}$ is essentially self-adjoint for $m\neq 0$,
  $\tilde{H}_{0}$ has deficiency indices $(1,1)$.
\end{prop}

\begin{proof}
  The operator $\tilde{H}_{m}$ is symmetric and semibounded, and so
  the deficiency indices are equal. If we set
  \begin{displaymath}
    \mu=\abs{m},\textrm{~}4\theta=-\frac{a^4\omega^2}{4}\,,
    \textrm{~}\lambda=-z-\frac{1}{4}\,,
  \end{displaymath}
  then the eigenvalue equation
  \begin{equation}
    \label{eigeneq}
    \left(-\parc{}{\xi}\left((\xi^{2}-1)\,\parc{}{\xi}\right)
      +\frac{m^2}{\xi^2-1}+\frac{a^{4}\omega^{2}}{4}\,(\xi^{2}-1)
      -\frac{1}{4}\right)\!\psi = z\psi
  \end{equation}
  takes the standard form of the differential equation of spheroidal
  functions \eqref{spheroidal}. According to chapter~3.12, Satz~5 in
  \cite{me}, for $\mu=m\in\N_{0}$ a fundamental system
  $\lsz{}y_{\text{I}},\ y_{\text{II}}\psz$ of solutions to equation
  \eqref{eigeneq} exists such that
  \begin{align*}
    & y_{\text{I}}(\xi) = (1-\xi)^{m/2}\,\mathfrak{P}_{1}(1-\xi),
    \quad\mathfrak{P}_{1}(0)=1, \\
    & y_{\text{II}}(\xi)=(1-\xi)^{-m/2}\mathfrak{P}_{2}(1-\xi)
    +A_{m}\,y_{\text{I}}(\xi)\log{(1-\xi)},
  \end{align*}
  where, for $\abs{\xi-1}<2$, $\mathfrak{P}_{1},\mathfrak{P}_{2}$ are
  analytic functions in $\xi$, $\lambda$, $\theta$; and $A_{m}$ is a
  polynomial in $\lambda$ and $\theta$ of total order $m$ with respect
  to $\lambda$ and $\sqrt{\theta}$; $A_{0}=-1/2$.

  Suppose that $z\in\C\setminus\R$. For $m=0$, every solutions to
  \eqref{eigeneq} is square integrable near $1$; while for $m\neq0$,
  $y_{\text{I}}$ is the only one solution, up to a factor, which is
  square integrable in a neighborhood of $1$. On the other hand, by a
  classical analysis due to Weyl, there exists exactly one linearly
  independent solution to \eqref{eigeneq} which is square integrable
  in a neighborhood of $\infty$, see Theorem~XIII.6.14 in \cite{ds}.
  In the case of $m=0$ this obviously implies that the deficiency
  indices are $(1,1)$. If $m\neq0$ then, by Theorem~XIII.2.30 in
  \cite{ds}, the operator $\tilde{H}_{m}$ is essentially self-adjoint.
\end{proof}

Define the maximal operator associated to the formal differential
expression
\begin{displaymath}
  L = -\parc{}{\xi}\left( (\xi^{2}-1)\parc{}{\xi}\right)
  +\frac{a^{4}\omega^{2}}{4}(\xi^{2}-1) -\frac{1}{4}
\end{displaymath}
as follows
\begin{equation*}
  \begin{split}
    & \begin{split} \Dom(H_{max}) = \bigg\{& f\in
      L^{2}((1,\infty),\ud\xi):\textrm{~}f,f'\in AC((1,\infty)), \\
      & -\parc{}{\xi}\left((\xi^{2}-1)\,\parc{f}{\xi}\right)
      +\frac{a^{4}\omega^{2}}{4}\,(\xi^{2}-1)
      f\in L^{2}((1,\infty),\ud\xi)\bigg\}, \\
    \end{split} \\
    & H_{max}f=Lf.
  \end{split}
\end{equation*}
According to Theorem~8.22 in \cite{weidmann},
$H_{max}=\tilde{H}_{0}^{\dagger}$.

\begin{prop}
  Let $\kappa\in(-\infty,\infty\,]$. The operator
  $\tilde{H}_{0}(\kappa)$ defined by the formulae
\begin{equation*}
  \begin{split}
    & \Dom(\tilde{H}_{0}(\kappa))
    = \lsz f\in\Dom(H_{max}):\textrm{~}
    f_{1}=\kappa f_{0}\psz,\textrm{~}
    \tilde{H}_{0}(\kappa)f=H_{max}f,
  \end{split}
\end{equation*}
where
\begin{displaymath}
  f_{0} := -4\pi{}a^{2}\lim_{\xi\to1+}\,
  \frac{f(\xi)}{\log{\!\left(2a^2(\xi-1)\right)}}\,,
  \textrm{~}f_{1} : =\lim_{\xi\to 1+}f(\xi)
  +\frac{1}{4\pi a^{2}}\,f_{0}\log{\!\left(2a^2(\xi-1\right))},
\end{displaymath}
is a self-adjoint extension of $\tilde{H}_{0}$. There are no other
self-adjoint extensions of $\tilde{H}_{0}$.
\end{prop}

\begin{proof}
  The methods to treat $\delta$ like potentials are now well
  established \cite{aghh}. Here we follow an approach described in
  \cite{bgp}, and we refer to this source also for the terminology and
  notations. Near the point $\xi=1$, each $f\in\Dom(H_{max})$ has the
  asymptotic behavior
  \begin{equation*}
    f(\xi) = f_{0}\,F(\xi,1)+f_{1}+o(1)\quad\textrm{as~}\xi\to1+
  \end{equation*}
  where $f_{0},f_{1}\in\C$ and $F(\xi,\xi')$ is the divergent part of
  the Green function for the Friedrichs extension of $\tilde{H}_{0}$.
  By formula \eqref{divpart} which is derived below,
  $F(\xi,1)=-1/(4\pi{}a^2)\log\!\left(2a^2(\xi-1)\right)$.
  Proposition~1.37 in \cite{bgp} states that
  $(\C,\Gamma_{1},\Gamma_{2})$, with $\Gamma_{1}f=f_{0}$ and
  $\Gamma_{2}f=f_{1}$, is a boundary triple for $H_{max}$.

  According to theorem 1.12 in \cite{bgp}, there is a one-to-one
  correspondence between all self-adjoint linear relations $\kappa$ in
  $\C$ and all self-adjoint extensions of $\tilde{H}_{0}$ given by
  $\kappa\longleftrightarrow\tilde{H}_{0}(\kappa)$ where
  $\tilde{H}_{0}(\kappa)$ is the restriction of $H_{max}$ to the
  domain of vectors $f\in\Dom(H_{max})$ satisfying
  \begin{equation}\label{boundarycon}
    (\Gamma_{1}f,\Gamma_{2}f)\in\kappa.
  \end{equation}
  Every self-adjoint relation in $\C$ is of the form
  $\kappa=\C{}v\subset\C^2$ for some $v\in\R^2$, $v\neq0$. If (with
  some abuse of notation) $v=(1,\kappa)$, $\kappa\in\R$, then relation
  \eqref{boundarycon} means that $f_{1}=\kappa{}f_{0}$.  If $v=(0,1)$
  then \eqref{boundarycon} means that $f_{0}=0$ which may be
  identified with the case $\kappa=\infty$.
\end{proof}

\begin{rem*}
  Let $\mathfrak{q}_0$ be the closure of the quadratic form associated
  to the semibounded symmetric operator $\tilde{H}_{0}$. Only the
  self-adjoint extension $\tilde{H}_{0}(\infty)$ has the property that
  all functions from its domain have no singularity at the point
  $\xi=1$ and belong to the form domain of $\mathfrak{q}_0$. It
  follows that $\tilde{H}_{0}(\infty)$ is the Friedrichs extension of
  $\tilde{H}_{0}$ (see, for example, Theorem~X.23 in \cite{rs} or
  Theorems~5.34 and 5.38 in \cite{weidmann}).
\end{rem*}

\subsection{The Green function}

Let us consider the Friedrichs extension of the operator $\tilde{H}$
in $L^{2}\left((1,\infty)\times{}S^1,\ud\xi\,\ud\phi\right)$ which was
introduced in (\ref{Htilde}). The resulting self-adjoint operator is
in fact the Hamiltonian for the impurity free case.  The corresponding
Green function $\Gz$ is the generalized kernel of the Hamiltonian, and
it should obey the equation
\begin{equation*}
  (\tilde{H}-z)\Gz (\xi,\phi;\xi',\phi')
  =\delta(\xi-\xi')\delta(\phi-\phi')=\frac{1}{2\pi}
  \sum_{m=-\infty}^{\infty}\delta(\xi-\xi')\ena{im(\phi-\phi')}.
\end{equation*}
If we suppose $\Gz$ to be of the form
\begin{equation}\label{decompG}
  \Gz (\xi,\phi;\xi',\phi')=\frac{1}{2\pi}
  \sum_{m=-\infty}^{\infty}\Gz^{m}(\xi,\xi')\ena{im(\phi-\phi')},
\end{equation}
then, for all $m\in\Z$,
\begin{equation}\label{partialgreen}
  (\tilde{H}_{m}-z)\Gz^{m}(\xi,\xi')=\delta(\xi-\xi').
\end{equation}

Let us consider an arbitrary fixed $\xi'$, and set
\begin{displaymath}
  \mu=m,\textrm{~}4\theta=-\frac{a^4\omega^2}{4},\textrm{~}
  \lambda=-z-\frac{1}{4}.
\end{displaymath}
Then for all $\xi\neq\xi'$ equation \eqref{partialgreen} takes the
standard form of the differential equation of spheroidal functions
\eqref{spheroidal}. As one can see from \eqref{asyS3}, the solution
which is square integrable near infinity equals
$S^{\abs{m}(3)}_{\nu}(\xi,-a^4\omega^2/16)$. Furthermore, the solution
which is square integrable near $\xi=1$ equals
$Ps^{\abs{m}}_{\nu}(\xi,-a^4\omega^2/16)$ as one may verify with the
aid of the asymptotic formula
\begin{equation*}
  P^{m}_{\nu}(\xi)\sim
  \frac{\Gamma(\nu+m+1)}{2^{m/2}\,m!\,\Gamma(\nu-m+1)}\,(\xi-1)^{m/2}
  \quad\textrm{as~}\xi\to1+,\textrm{~for~}m\in\N_{0}.
\end{equation*}
 
We conclude that the $m$th partial Green function equals
\begin{equation}
  \label{Greenm}
  \Gz^{m}(\xi,\xi') = -\frac{1}{(\xi^2-1)
    \mathscr{W}(Ps^{\abs{m}}_{\nu},S^{\abs{m}(3)}_{\nu})}\,
  Ps^{\abs{m}}_{\nu}\!\left(\xi_{<},-\frac{a^4\omega^2}{16}\right)
  S^{\abs{m}(3)}_{\nu}\!\left(\xi_{>},-\frac{a^4\omega^2}{16}\right)
\end{equation}
where the symbol
$\mathscr{W}(Ps^{\abs{m}}_{\nu},S^{\abs{m}(3)}_{\nu})$ denotes the
Wronskian, and $\xi_{<}$, $\xi_{>}$ are respectively the smaller and
the greater of $\xi$ and $\xi'$. By the general Sturm-Liouville
theory, the factor
$(\xi^2-1)\mathscr{W}(Ps^{\abs{m}}_{\nu},S^{\abs{m}(3)}_{\nu})$ is
constant. Since $\Gz^{m}=\Gz^{-m}$ decomposition $\eqref{decompG}$ may
be simplified,
\begin{equation}
  \label{Green_cos}
  \Gz (\xi,\phi;\xi',\phi')
  = \frac{1}{2\pi}\,\Gz^{0}(\xi,\xi')
  +\frac{1}{\pi}\sum_{m=1}^{\infty}\Gz^{m}(\xi,\xi')\cos{[m(\phi-\phi')]}.
\end{equation}

\subsection{The Krein $Q$-function}

The Krein $Q$-function plays a crucial role in the spectral analysis
of impurities. It is defined at a point of the configuration space as
the regularized Green function evaluated at this point. Here we deal
with the impurity located in the center of the dot ($\xi=1$, $\phi$
arbitrary), and so, by definition,
\begin{displaymath}
  Q(z):=\mathcal{G}^{reg}_z(1,0;1,0).
\end{displaymath}

Due to the rotational symmetry,
\begin{displaymath}
  \Gz(\xi) := \Gz(\xi,\phi;1,0) = \Gz(\xi,\phi;1,\phi)
  = \Gz(\xi,0;1,0) = \frac{1}{2\pi}\,\Gz^{0}(\xi,1),
\end{displaymath}
and hence
\begin{equation*}
  (\tilde{H}_{0}-z)\Gz(\xi)=0,\quad\textrm{for~}\xi\in(1,\infty).
\end{equation*}
Let us note that from the explicit formula (\ref{Greenm}), one can
deduce that the coefficients $\Gz^{m}(\xi,1)$ in the series in
(\ref{Green_cos}) vanish for $m=1,2,3,\ldots$. The solution to this
equation is
\begin{equation*}
  \Gz(\xi) 
  \propto S^{0(3)}_{\nu}\!\left(\xi,-\frac{a^4\omega^2}{16}\right).
\end{equation*}
The constant of proportionality can be determined with the aid the
following theorem which we reproduce from \cite{ondiag}.

\begin{thm}
  \label{diag}
  Let $d(x,y)$ denote the geodesic distance between points $x,y$ of a
  two-dimensional manifold $X$ of bounded geometry. Let
  \begin{displaymath}
    U\in\mathcal{P}(X):=\Big\{U:\textrm{~}
    U_{+}:=\max(U,0)\in L^{p_{0}}_{loc}(X),\ U_{-}:=\max(-U,0)\in
    \sum_{i=1}^{n}L^{p_{i}}(X) \Big\}
  \end{displaymath}
  for an arbitrary $n\in\N$ and $2\leq p_{i}\leq\infty$. Then the
  Green function $\mathcal{G}_{U}$ of the Schr\"odinger operator
  $H_{U}=-\Delta_{LB}+U$ has the same on-diagonal singularity as that
  for the Laplace-Beltrami operator itself, i.e.,
  \begin{equation*}
    \mathcal{G}_{U}(\zeta;x,y) = \frac{1}{2\pi}\log\frac{1}{d(x,y)}
    +\mathcal{G}_{U}^{reg}(\zeta;x,y)
  \end{equation*}
  where $\mathcal{G}_{U}^{reg}$ is continuous on $X\times{}X$.
\end{thm}

Let us denote by $\Gz^{H}$ and $Q^{H}(z)$ the Green function and the
Krein $Q$-function for the Friedrichs extension of $H$, respectively.
Since $\tilde{H}=a^2H$ and $(\tilde{H}-z)\Gz=\delta$, we have
\begin{equation*}
  \Gz^{H}(\xi,\phi;\xi',\phi')
  =a^2\mathcal{G}_{a^2 z}(\xi,\phi;\xi',\phi'),\textrm{~}
  Q^{H}(z)=a^{2}Q(a^{2}z).
\end{equation*}
One may verify that
\begin{equation*}
  \log d(\varrho,0;\vec{0})=\log\varrho=\log (a\arg\cosh\xi)
  =\frac{1}{2}\log\!\left(2a^2(\xi-1)\right)+O(\xi-1)
\end{equation*}
as $\varrho\to0+$ or, equivalently, $\xi\to1+$. Finally, for the
divergent part
\begin{displaymath}
  F(\xi,\xi') := \Gz(\xi,\phi;\xi',\phi)-\Gz^{reg}(\xi,\phi;\xi',\phi)
  = \Gz(\xi,0;\xi',0)-\Gz^{reg}(\xi,0;\xi',0)
\end{displaymath}
of the Green function $\Gz$ we obtain the expression
\begin{equation}
  \label{divpart}
  F(\xi,1)=-\frac{1}{4\pi a^{2}}\,\log\!\left(2a^2(\xi-1)\right).
\end{equation}
From the above discussion, it follows that the Krein $Q$-function
depends on the coefficients $\alpha$, $\beta$ in the asymptotic
expansion
\begin{equation}\label{asysol}
  S^{0(3)}_{\nu}\!\left(\xi,-\frac{a^4\omega^2}{16}\right)
  = \alpha\log(\xi-1)+\beta+o(1)\quad\textrm{as~}\xi\to1+,
\end{equation}
and equals
\begin{equation}
  \label{KreinQgen}
  Q(z) = -\frac{\beta}{4\pi a^{2}\alpha}
  +\frac{\log(2a^{2})}{4\pi a^{2}}\,.
\end{equation}

To determine $\alpha$, $\beta$ we need relation (\ref{eq:S3}) for the
radial spheroidal function of the third kind. For $\nu$ and $\nu+1/2$
being non-integer, formula \eqref{joining} implies that
\begin{equation}
  \label{radialasy}
  \begin{split} 
    & S^{0(1)}_{\nu}(\xi,\theta)
    = \frac{\sin(\nu\pi)}{\pi}\,
    \ena{-i\pi(\nu+1)}K^{0}_{\nu}(\theta)
    Qs^{0}_{-\nu-1}(\xi,\theta), \\
    & S^{0(1)}_{-\nu-1}(\xi,\theta)
    = \frac{\sin(\nu\pi)}{\pi}\,
    \ena{i\pi\nu}K^{0}_{-\nu-1}(\theta)Qs^{0}_{\nu}(\xi,\theta).
  \end{split}
\end{equation}
Applying the symmetry relation \eqref{coeffsym} for expansion
coefficients, we derive that
\begin{equation*}
  \begin{split}
    Qs^{0}_{-\nu-1}(\xi,\theta)&
    = \sum_{r=-\infty}^{\infty}(-1)^{r}a^{0}_{-\nu-1,r}(\theta)
    Q^{0}_{-\nu-1+2r}(\xi)\\
    & = \sum_{r=-\infty}^{\infty}(-1)^{r}a^{0}_{\nu,r}(\theta)
    Q^{0}_{-\nu-1-2r}(\xi).\end{split}
\end{equation*}
Using the asymptotic formulae (see \cite{be})
\begin{equation*}
  Q^{0}_{\nu}(\xi) = -\frac{1}{2}\,
  \log\frac{\xi-1}{2}+\Psi(1)-\Psi(\nu+1)
  +O\!\left((\xi-1)\log(\xi-1)\right),
\end{equation*}
the series expansion in \eqref{angularexpansion} and formulae
\eqref{radialasy}, we deduce that, as $\xi\to1+$,
\begin{equation*}
  \begin{split}
    & \begin{split}
      S^{0(1)}_{\nu}(\xi,\theta) \sim &
      -\frac{\sin(\nu\pi)}{\pi}\,
      \ena{-i\pi(\nu+1)}K^{0}_{\nu}(\theta) \\
      & \times\left[ s^{0}_{\nu}(\theta)^{-1}
        \left(\frac{1}{2}
          \log\frac{\xi-1}{2}-\Psi(1)+\pi\cot(\nu\pi)\right)
        +\Psi{}s_{\nu}(\theta)\right],
    \end{split} \\
    & \begin{split}
      S^{0(1)}_{-\nu-1}(\xi,\theta) \sim &
      -\frac{\sin(\nu\pi)}{\pi}\,
      \ena{i\pi\nu}K^{0}_{-\nu-1}(\theta) \\
      & \times\left[ s^{0}_{\nu}(\theta)^{-1}\left(\frac{1}{2}
          \log\frac{\xi-1}{2}-\Psi(1)\right)
        +\Psi{}s_{\nu}(\theta)\right],
    \end{split} 
  \end{split}
\end{equation*}
where the coefficients $s^{\mu}_{n}(\theta)$ are introduced in
(\ref{eq:coeff_s}),
\begin{displaymath}
  \Psi{}s_{\nu}(\theta) := \sum^{\infty}_{r=-\infty}(-1)^{r}
  a^{0}_{\nu,r}(\theta)\Psi(\nu+1+2r),
\end{displaymath}
and where we have made use of the following relation for the digamma
function: $\Psi(-z)=\Psi(z+1)+\pi\cot(\pi{}z)$.

We conclude that
\begin{equation*}
  S^{0(3)}_{\nu}(\xi,\theta)
  \sim \alpha\log(\xi-1)+\beta+ O\left((\xi-1)\log(\xi-1)\right)
  \quad\textrm{as~}\xi\to1+,
\end{equation*}
where
\begin{align*}
  & \alpha=\frac{i\tan(\nu\pi)}{2\pi s^{0}_{\nu}(\theta)}
  \left(\ena{i\pi\nu} K^{0}_{-\nu-1}(\theta)-\ena{-i\pi(2\nu+3/2)}
    K^{0}_{\nu}(\theta)\right),\\
  &\beta=\alpha\left(-\log 2-2\Psi(1)
    +2\Psi s_{\nu}(\theta)s^{0}_{\nu}(\theta)\right)
  + \ena{-2i\pi\nu} s^{0}_{\nu}(\theta)^{-1}K^{0}_{\nu}(\theta).
\end{align*}
The substitution for $\alpha$, $\beta$ into \eqref{KreinQgen} yields
\begin{equation}
  \label{Qfunction}
  \begin{split}
    Q(z) = & -\frac{1}{4\pi a^{2}}\left(-\log 2-2\Psi(1)
      +2\,\Psi s_{\nu}\!\left(-\frac{a^4\omega^2}{16}\right)
      s^{0}_{\nu}\!\left(-\frac{a^4\omega^2}{16}\right)\right)\\
    & +\frac{1}{2 a^{2}\tan(\nu\pi)}\left(\ena{i\pi(3\nu+3/2)}\,
      \frac{K^{0}_{-\nu-1}(-\frac{a^4\omega^2}{16})}{
        K^{0}_{\nu}(-\frac{a^4\omega^2}{16})}-1\right)^{\!-1}
    +\frac{\log(2a^{2})}{4\pi a^{2}}
  \end{split}
\end{equation}
where $\nu$ is chosen so that
\begin{equation}\label{specpar}
  \lambda^{0}_{\nu}\!\left(-\frac{a^4\omega^2}{16}\right)
  = -z-\frac{1}{4}\,.
\end{equation}

For $\nu=n$ being an integer, we can immediately use the known
asymptotic formulae for spheroidal functions (see Section~16.12 in
\cite{be}) which yield
\begin{equation*}
  \begin{split}
    S^{0(3)}_{n}(\xi,\theta)
    =\,& \frac{i s^{0}_{n}(\theta)}{4\sqrt{\theta}
      K^{0}_{n}(\theta)}\,\log(\xi-1)
    -\frac{i s^{0}_{n}(\theta)\log{2}}{4\sqrt{\theta}
      K^{0}_{n}(\theta)} \\
    & +\frac{i s^{0}_{n}(\theta)^{2}}{2\sqrt{\theta}
      K^{0}_{n}(\theta)}\,
    \sum_{2r\geq -n}(-1)^{r}a^{0}_{n,r}(\theta)h_{n+2r}
    +\frac{K_{n}^{0}(\theta)}{s^{0}_{n}(\theta)}+O(\xi-1),
  \end{split}
\end{equation*}
as $\xi\to1+$. Here, $h_{0}=1,h_{k}=1/1+1/2+\ldots+1/k$. By
\eqref{KreinQgen}, one can calculate the $Q$-function in this case,
too.

\subsection{The spectrum of a quantum dot with impurity}

The Green function of the Hamiltonian describing a quantum dot with
impurity is given by the Krein resolvent formula
\begin{equation*}
  \Gz^{H(\chi)}(\xi,\phi;\xi',\phi')
  = \Gz^{H}(\xi,\phi;\xi',\phi')-\frac{1}{Q^H(z)-\chi}\,
  \Gz^{H}(\xi,0;1,0)\Gz^{H}(1,0;\xi',0)
\end{equation*}
(recall that, due to the rotational symmetry,
$\Gz^{H}(\xi,\phi;1,0)=\Gz^{H}(\xi,0;1,0)$). The parameter
$\chi:=a^2\kappa\in (-\infty,\infty\,]$ determines the corresponding
self-adjoint extension $H(\chi)$ of $H$. In the physical
interpretation, this parameter is related to the strength of the
$\delta$ interaction. Recall that the value $\chi=\infty$ corresponds
to the Friedrichs extension of $H$ representing the case with no
impurity. This fact is also apparent from the Krein resolvent formula.

The unperturbed Hamiltonian $H(\infty)$ describes a harmonic
oscillator on the Lobachevsky plane. As is well known (see, for
example, \cite{bs}), for the confinement potential tends to infinity
as $\varrho\to\infty$, the resolvent of $H(\infty)$ is compact, and
the spectrum of $H(\infty)$ is discrete and semibounded. The
eigenvalues of $H(\infty)$ are solutions of a scalar equation whose
introduction also relies heavily on the theory of spheroidal functions.
We are sceptic about the possibility of deriving an explicit formula
for the eigenvalues. But the equation turned out to be convenient
enough to allow for numerical solutions. A more detailed discussion
jointly with a basic numerical analysis is provided in a separate
paper \cite{ts}.

A similar observation about the basic spectral properties
(discreteness and semiboundedness) is also true for the operators
$H(\chi)$ for any $\chi\in\R$ since, by the Krein resolvent formula,
the resolvents for $H(\chi)$ and $H(\infty)$ differ by a rank one
operator. Moreover, the multiplicities of eigenvalues of $H(\chi)$ and
$H(\infty)$ may differ at most by $\pm1$ (see
\cite[Section~8.3]{weidmann}).

A more detailed and rather general analysis which is given in
\cite{bgl} can be carried over to our case almost literally. Denote by
$\sigma$ the set of poles of the function $Q^H(z)$ depending on the
spectral parameter $z$. Note that $\sigma$ is a subset of
$\spec(H(\infty))$. Consider the equation
\begin{equation}
  \label{Qkappa}
  Q^H(z) = \chi.
\end{equation}

\begin{thm}
  The spectrum of $H(\chi)$ is discrete and consists of four
  nonintersecting parts $S_{1}$, $S_{2}$, $S_{3}$, $S_{4}$ described
  as follows:
  \begin{enumerate}
  \item $S_{1}$ is the set of all solutions to equation (\ref{Qkappa})
    which do not belong to the spectrum of $H(\infty)$. The
    multiplicity of all these eigenvalues in the spectrum of $H(\chi)$
    equals 1.
  \item $S_{2}$ is the set of all $\lambda\in\sigma$ that are multiple
    eigenvalues of $H(\infty)$. If the multiplicity of such an
    eigenvalue $\lambda$ in $\spec(H(\infty))$ equals $k$ then its
    multiplicity in the spectrum of $H(\chi)$ equals $k-1$.
  \item $S_{3}$ consists of all
    $\lambda\in\spec(H(\infty))\setminus\sigma$ that are not solutions
    to equation (\ref{Qkappa}). the multiplicities of such an
    eigenvalue $\lambda$ in $\spec(H(\infty))$ and $\spec(H(\chi))$
    are equal.
  \item $S_{4}$ consists of all
    $\lambda\in\spec(H(\infty))\setminus\sigma$ that are solutions to
    equation (\ref{Qkappa}). If the multiplicity of such an eigenvalue
    $\lambda$ in $\spec(H(\infty))$ equals $k$ then its multiplicity
    in the spectrum of $H(\chi)$ equals $k+1$.
  \end{enumerate}
\end{thm}

Hence the eigenvalues of $H(\chi)$, $\chi\in\R$, different from those
of the unperturbed Hamiltonian $H(\infty)$ are solutions to
(\ref{Qkappa}). As far as we see it, this equation can be solved only
numerically. We have postponed a systematic numerical analysis of
equation (\ref{Qkappa}) to a subsequent work. Note that the Krein
$Q$-function \eqref{Qfunction} is in fact a function of $\nu$, and
hence dependence \eqref{specpar} of the spectral parameter $z$ on
$\nu$ is fundamental.  In this context, it is quite useful to know for
which values of $\nu$ the spectral parameter $z$ is real. A partial
answer is given by Proposition~\ref{lambdapropo}.

\section{Conclusion}

We have proposed a Hamiltonian describing a quantum dot in the
Lobachevsky plane to which we added an impurity modeled by a $\delta$
potential. Formulas for the corresponding $Q$- and Green functions
have been derived. Further analysis of the energy spectrum may be
accomplished for some concrete values of the involved parameters (by
which we mean the curvature $a$ and the oscillator frequency $\omega$)
with the aid of numerical methods.

\setcounter{section}{1}
\renewcommand{\thesection}{\Alph{section}}
\setcounter{equation}{0}
\renewcommand{\theequation}{\Alph{section}.\arabic{equation}}

\section*{Appendix: Spheroidal functions}

Here we follow the source \cite{be}. Spheroidal functions are
solutions to the equation
\begin{equation}\label{spheroidal}
  (1-\xi^{2})\parc{^2\psi}{\xi^{2}}-2\xi\parc{\psi}{\xi}
  +\left[\lambda+4\theta(1-\xi^2)-\mu^2(1-\xi^2)^{-1} \right]\psi=0, 
\end{equation}
where all parameters are in general complex numbers. There are two
solutions that behave like $\xi^{\nu}$ times a single-valued function
and $\xi^{-\nu-1}$ times a single-valued function at $\infty$. The
exponent $\nu$ is a function of $\lambda$, $\theta$, $\mu$, and is
called the characteristic exponent. Usually, it is more convenient to
regard $\lambda$ as a function of $\nu$, $\mu$ and $\theta$. We shall
write $\lambda=\lambda^{\mu}_{\nu}(\theta)$. If $\nu$ or $\mu$ is an
integer we denote it by $n$ or $m$, respectively. The functions
$\lambda^{\mu}_{\nu}(\theta)$ obey the symmetry relations
\begin{equation}
  \label{eq:lambdasym}
  \lambda^{\mu}_{\nu}(\theta)
  = \lambda^{-\mu}_{\nu}(\theta)
  = \lambda^{\mu}_{-\nu-1}(\theta) 
  = \lambda^{-\mu}_{-\nu-1}(\theta).
\end{equation}

A first group of solutions (radial spheroidal functions) is obtained
as expansions in series of Bessel functions,
\begin{equation}
  \label{spheroidalS}
  S^{\mu (j)}_{\nu}(\xi,\theta)=(1-\xi^{-2})^{-\mu/2}
  s^{\mu}_{\nu}(\theta)\sum_{r=-\infty}^{\infty}a^{\mu}_{\nu,r}
  \psi^{(j)}_{\nu +2r}(2\theta^{1/2}\xi),
\end{equation}
$j=1,2,3,4$, where the factors $s^{\mu}_{\nu}(\theta)$ are determined
below and
\begin{align*}
  & \psi^{(1)}_{\nu}(\zeta) = \sqrt{\frac{\pi}{2\zeta}}\,
  J_{\nu +1/2}(\zeta),\quad
  \psi^{(2)}_{\nu}(\zeta)
  = \sqrt{\frac{\pi}{2\zeta}}\,Y_{\nu +1/2}(\zeta), \\
  & \psi^{(3)}_{\nu}(\zeta) = \sqrt{\frac{\pi}{2\zeta}}\,
  H^{(1)}_{\nu +1/2}(\zeta),\quad
  \psi^{(4)}_{\nu}(\zeta)
  = \sqrt{\frac{\pi}{2\zeta}}\,H^{(2)}_{\nu +1/2}(\zeta).
\end{align*}
The coefficients $a^{\mu}_{\nu,r}(\theta)$ (denoted only $a_{r}$ for
the sake of simplicity) satisfy a three term recurrence relation
\begin{equation}\label{rekurence}
  \begin{split}
    & \frac{(\nu +2r-\mu)(\nu +2r -\mu -1)}{(\nu +2r-3/2)
      (\nu +2r-1/2)}\,\theta a_{r-1}
    +\frac{(\nu +2r +\mu +2)(\nu +2r +\mu +1)}{(\nu +2r +3/2)
      (\nu +2r +5/2)}\,\theta a_{r+1}\\
    & +\left[ \lambda^{\mu}_{\nu}(\theta)
      -(\nu +2r)(\nu +2r+1)+\frac{(\nu +2r)(\nu +2r+1)+\mu^2-1}
      {(\nu +2r-1/2)(\nu +2r +3/2)}\,2\theta \right]a_{r}=0.
  \end{split}
\end{equation}
Here and in what follows we assume that $\nu +1/2$ is not an integer
(to our knowledge, the omitted case is not yet fully investigated).

The coefficients $a^{\mu}_{\nu,r}(\theta)$ may be chosen such that
\begin{equation*}
  a^{\mu}_{\nu,0}(\theta)=a^{\mu}_{-\nu -1,0}(\theta)
  =a^{-\mu}_{\nu,0}(\theta),
\end{equation*}
and so (see (\ref{eq:lambdasym}))
\begin{equation}
  \label{coeffsym}
  a^{\mu}_{\nu,r}(\theta)=a^{\mu}_{-\nu-1,-r}(\theta)
  = \frac{(\nu-\mu+1)_{2r}}{(\nu+\mu+1)_{2r}}\,
  a^{-\mu}_{\nu,r}(\theta)
\end{equation}
where $(a)_{r}:=a(a+1)(a+2)\dots(a+r-1)=\Gamma(a+r)/\Gamma(a)$,
$(a)_{0}:=1$. Equation \eqref{rekurence} leads to a convergent
infinite continued fraction and this way one can prove that
\begin{equation}\label{coefasy}
  \lim_{r\to\infty}\,\frac{r^2 a_{r}}{a_{r-1}}
  = \lim_{r\to -\infty}\,\frac{r^2 a_{r}}{a_{r+1}}=\frac{\theta}{4}.
\end{equation}
From \eqref{coefasy} and the asymptotic formulae for Bessel functions,
it follows that \eqref{spheroidalS} converges if $\abs{\xi}>1$.

If we set in (\ref{spheroidalS})
\begin{equation}
  \label{eq:coeff_s}
  s^{\mu}_{\nu}(\theta)
  = \left[\sum_{r=-\infty}^{\infty}
    (-1)^{r}a^{\mu}_{\nu, r}(\theta)\right]^{-1}
\end{equation}
then
\begin{equation*}
  S^{\mu (j)}_{\nu}(\xi,\theta)\sim\psi^{(j)}_{\nu}(2\theta^{1/2}\xi),
  \quad\textrm{for~}
  \abs{\arg (\theta ^{1/2}\xi)} < \pi\quad\textrm{as }\xi\to\infty.
\end{equation*}
We have the asymptotic forms, valid as $\xi\to\infty$,
\begin{align}
  && \begin{split} 
    \label{asyS3}
    & S^{\mu(3)}_{\nu}(\xi,\theta)
    = \frac{1}{2}\,\theta^{-1/2}\xi^{-1}\ena{i(2\theta^{1/2}
      \xi-\nu\pi/2-\pi/2)}[1+O(\abs{\xi}^{-1})],\\
    & \textrm{for~} -\pi<\arg (\theta^{1/2}\xi)<2\pi,
  \end{split}
\end{align}
and
\begin{align}
  && \begin{split}
    \label{asyS4}
    & S^{\mu(4)}_{\nu}(\xi,\theta)
    = \frac{1}{2}\,\theta^{-1/2}\xi^{-1}
    \ena{-i(2\theta^{1/2}\xi-\nu\pi/2-\pi/2)}[1+O(\abs{\xi}^{-1})],\\
    & \textrm{for~} -2\pi<\arg (\theta^{1/2}\xi)<\pi.
  \end{split}
\end{align}

The radial spheroidal functions satisfy the relation
\begin{equation}
  \label{eq:S3}
  S^{\mu(3)}_{\nu} = \frac{1}{i\cos(\nu\pi)}
  \left(S^{\mu(1)}_{-\nu-1}+i\,\ena{-i\pi\nu}S^{\mu(1)}_{\nu}\right).
\end{equation}

The radial spheroidal functions are especially useful for large $\xi$;
the larger is the $\xi$ the better is the convergence of the
expansion. To obtain solutions useful near $\pm1$, and even on the
segment $(-1,1)$, one uses expansions in series in Legendre functions,
\begin{equation}
  \label{angularexpansion}
  \begin{split}
    & Ps^{\mu}_{\nu}(\xi,\theta)
    = \sum_{r=-\infty}^{\infty}(-1)^{r}a^{\mu}_{\nu,r}(\theta)
    P^{\mu}_{\nu +2r}(\xi), \\
    & Qs^{\mu}_{\nu}(\xi,\theta)
    = \sum_{r=-\infty}^{\infty}(-1)^{r}a^{\mu}_{\nu,r}(\theta)
    Q^{\mu}_{\nu +2r}(\xi).
  \end{split}
\end{equation}
These solutions are called the angular spheroidal functions and are
related to the radial spheroidal functions by the following formulae:
\begin{equation}\label{joining}
  \begin{split}
    & S^{\mu(1)}_{\nu}(\xi,\theta)
    = \pi^{-1}\sin[(\nu-\mu)\pi]\ena{-i\pi(\nu+\mu+1)}
    K^{\mu}_{\nu}(\theta)Qs^{\mu}_{-\nu-1}(\xi,\theta), \\
    & S^{m(1)}_{n}(\xi,\theta)
    = K^{m}_{n}(\theta)Ps^{m}_{n}(\xi,\theta),
  \end{split}
\end{equation}
where $K^{\mu}_{\nu}(\theta)$ can be expressed as a series in
coefficients $a^{\mu}_{\nu,r}(\theta)$, and sometimes it is called the
joining factor. In more detail, for any $k\in\Z$ it holds true that
\begin{equation*}
  \begin{split}
    K^{\mu}_{\nu}(\theta)
    =\, & \frac{1}{2}\left(\frac{\theta}{4}\right)^{\!\nu/2+k}
    \Gamma(1+\nu-\mu+2k)\,\ena{(\nu+k)\pi i}s^\mu_\nu(\theta)
    \\
    & \times\frac{\displaystyle\sum_{r=-\infty}^k
      \frac{(-1)^ra^\mu_{\nu,r}(\theta)}{(k-r)!\,\Gamma(\nu+k+r+3/2)}}
    {\displaystyle\sum_{r=k}^\infty
      \frac{(-1)^ra^\mu_{\nu,r}(\theta)}{(r-k)!\,\Gamma(1/2-\nu-k-r)}}
    \,.
  \end{split}
\end{equation*}

\begin{prop}\label{lambdapropo}
  Let $\nu,\theta\in\R$ and set $\mu =0$. Then
  $\lambda^{0}_{\nu}(\theta)\in\R$.
\end{prop}

\begin{proof}
  To simplify the notation we denote, in \eqref{rekurence},
  \begin{align*}
    & \alpha^{\mu,\nu}_{r}(\theta)
    = \frac{(\nu +2r+\mu +2)(\nu +2r+\mu +1)}{(\nu +2r+3/2)
      (\nu+2r+5/2)}\,\theta, \\
    & \beta^{\mu,\nu}_{r}(\theta)
    = -(\nu +2r)(\nu +2r+1)+\frac{(\nu+2r)(\nu+2r+1)+\mu^2-1}
    {(\nu +2r-1/2)(\nu +2r +3/2)}\,2\theta, \\
    & \gamma^{\mu,\nu}_{r}(\theta) = \frac{(\nu +2r-\mu)
      (\nu +2r-\mu -1)}{(\nu +2r-3/2)(\nu +2r-1/2)}\,\theta.
  \end{align*}
  The resulting formula may be written in the matrix form,
  \begin{equation}\label{rekurencematrix}
    \begin{pmatrix}
      &\ddots  &  &  &  &  &  &\\
      &  &\gamma_{-1}  &\beta_{-1}  &\alpha_{-1}  &  & &  \\
      &  &  &\gamma_{0}  &\beta_{0}  &\alpha_{0}  & & \\
      &  &  &  &\gamma_{1}  &\beta_{1}  &\alpha_{1}&  \\
      &  &  &  &  &  & & \ddots 
    \end{pmatrix}
    \begin{pmatrix}
      \vdots \\
      a_{-1}\\
      a_{0}\\
      a_{1}\\
      \vdots 
    \end{pmatrix}=
    -\lambda
    \begin{pmatrix}
      \vdots \\
      a_{-1}\\
      a_{0}\\
      a_{1}\\
      \vdots 
    \end{pmatrix}
  \end{equation}
  where we have omitted the fixed indices.

  As one can see,
  \begin{equation*}
    \gamma^{0, \nu}_{r+1}(\theta)
    = \frac{\nu+2r+5/2}{\nu+2r+1/2}\,\alpha^{0,\nu}_{r}(\theta)
  \end{equation*}
  and so
  \begin{equation*}
    \frac{\nu+2r+1/2}{\nu+2r-3/2}\alpha^{0,\nu}_{r-1}(\theta)
    a^{0}_{\nu,r-1}(\theta)+\beta^{0,\nu}_{r}
    (\theta)a^{0}_{\nu,r}(\theta)+\alpha^{0,\nu}_{r}(\theta)
    a^{0}_{\nu,r+1}(\theta)=-\lambda^{0}_{\nu}(\theta)
    a^{0}_{\nu,r}(\theta).
  \end{equation*}
  Substitution $a^{0}_{\nu,r}=L_{r}(\nu)\tilde{a}^{0}_{\nu,r}$, where
  $L_{r}(\nu)$ are non-zero constants, yields
  \begin{equation*}
    \begin{split}
      &\frac{\nu+2r+1/2}{\nu+2r-3/2}\alpha^{0,\nu}_{r-1}(\theta)
      \tilde{a}^{0}_{\nu,r-1}(\theta)\frac{L_{r-1}(\nu)}
      {L_{r}(\nu)}+\beta^{0,\nu}_{r}(\theta)
      \tilde{a}^{0}_{\nu,r}(\theta)+\alpha^{0,\nu}_{r}(\theta)
      \tilde{a}^{0}_{\nu,r+1}(\theta)\frac{L_{r+1}(\nu)}{L_{r}(\nu)}\\
      &=-\lambda^{0}_{\nu}(\theta)\tilde{a}^{0}_{\nu,r}(\theta).
    \end{split}
  \end{equation*}
  We require the matrix in \eqref{rekurencematrix} to be symmetric in
  the new coordinates $\lsz\tilde{a}_{r}\psz$. This implies that
  \begin{equation*}
    \frac{\nu+2r+1/2}{\nu+2r-3/2}\frac{L_{r-1}(\nu)}{L_{r}(\nu)}
    =\frac{L_{r}(\nu)}{L_{r-1}(\nu)}.
  \end{equation*}
  For $r\notin (-\nu/2-1/4,-\nu/2+3/4)$, the solution is
  $L_{r}(\nu)=\sqrt{\abs{\nu+2r+1/2}}$. For $r_{0}\equiv
  r\in(-\nu/2-1/4,-\nu/2+3/4)$, there is no real solution and so we
  set $L_{r_{0}}(\nu)=\sqrt{\abs{\nu+2r_{0}+1/2}}$ and make another
  transformation of coordinates:
  \begin{equation*}
    \tilde{\tilde{a}}_{r}=\begin{cases}
      -\tilde{a}_{r} & \quad\textrm{for~}r=r_{0}-(2k-1),
      \ k\in\N, \\
      \tilde{a}_{r} & \quad\textrm{for~all~other~}r.	
    \end{cases}
  \end{equation*}

  Relation \eqref{rekurencematrix} can be viewed as an eigenvalue
  equation with a symmetric matrix in the coordinate system $\lsz
  \tilde{\tilde{a}}_{k}\psz$, hence $\lambda^{0}_{\nu}(\theta)$ must
  be real.
\end{proof}

\subsection*{Acknowledgments}

The authors wish to acknowledge gratefully partial support from the
following grants: grant No.~201/05/0857 of the Grant Agency of Czech
Republic (P.~\v{S}.) and grant No.~LC06002 of the Ministry of
Education of Czech Republic (M.~T.).

\end{document}